\documentclass[conference]{IEEEtran}
\IEEEoverridecommandlockouts
\usepackage{hyperref}
\usepackage[top=0.75in, bottom=1.1 in, left=0.64in, right=0.64in]{geometry}
 \usepackage{lipsum,graphicx,multicol}
 \usepackage{comment}
\usepackage{color,soul}
\usepackage{mathtools}
\usepackage{graphicx}
\usepackage{balance}
\usepackage{color,soul}
\usepackage{setspace}
\usepackage{multirow}
\usepackage{tablefootnote}
\usepackage[norelsize, linesnumbered, ruled, lined, boxed, commentsnumbered]{algorithm2e}
\usepackage[normalem]{ulem}
\usepackage{comment}
\usepackage[font=small,labelfont=bf]{caption}
\usepackage{subcaption}
\usepackage{bm, amsmath, amssymb, amsthm}
\usepackage{amsfonts}
\usepackage {amssymb,amsmath,xcolor}
\usepackage{booktabs}
\usepackage{cite}
\usepackage{array}
\usepackage[acronym]{glossaries}
\usepackage{cases}
\usepackage{cleveref}
\usepackage{algorithmic}

\usepackage{amsthm}
\usepackage{array}
\usepackage{float}
\usepackage{graphicx}
\usepackage{epstopdf}

\def\BibTeX{{\rm B\kern-.05em{\sc i\kern-.025em b}\kern-.08em
    T\kern-.1667em\lower.7ex\hbox{E}\kern-.125emX}}
    
  \newtheorem{theorem}{Theorem}

\newtheorem{corollary}{Corollary}

\newtheorem{lemma}{Lemma}
\newtheorem{remark}{Remark}

\newcommand{\bieee}{\begin{IEEEeqnarray}{rCl}}
\newcommand{\eieee}{\end{IEEEeqnarray}}

\usepackage{amsthm,mdframed,calc}
 \usepackage{float}
\usepackage{diagbox}

\newcommand{\kim}[1]{{\color{purple} [Kim: #1]}}

\SetKwInput{KwInput}{Input}                
\SetKwInput{KwOutput}{Output}              

\title{Multiple-Function Approximated Coded Computing: Constructions and Error Bounds}

\author{\IEEEauthorblockN{
Rimpi Borah\IEEEauthorrefmark{1},
Wilton Kim\IEEEauthorrefmark{2},
J. Harshan\IEEEauthorrefmark{1},
Stanislav Kruglik\IEEEauthorrefmark{3},
Han Mao Kiah\IEEEauthorrefmark{2}}

\IEEEauthorblockA{\IEEEauthorrefmark{1}Indian Institute of Technology Delhi, India, \IEEEauthorrefmark{2}Nanyang Technological University, Singapore} \IEEEauthorblockA{\IEEEauthorrefmark{3}Technical University of Denmark, Denmark}}
\begin{document}
\maketitle

\begin{abstract}

Coded computing is well known for mitigating stragglers in distributed computation. Within the class of coded computing schemes, Lagrange coded computing (LCC) enables exact polynomial computation over finite fields, whereas Berrut approximated coded computing (BACC) supports approximate computation of arbitrary functions, including non-polynomial functions, in the analog domain. In applications such as multi-task learning and inference, computing multiple target functions over different subsets of a common dataset separately may require multiple communication rounds. In contrast, computing these functions jointly in a single communication round can reduce communication overhead while retaining redundancy against stragglers. An existing approach combines multiple functions into a joint function using polynomial masks, however, it is restricted to polynomial computation over finite fields, leaving single-round multi-function computation for arbitrary functions in the analog domain an open problem. To address this gap, we develop a BACC counterpart to this approach and derive an approximation-error bound in terms of the derivatives of the joint function. We show that the bound contains a factor that grows exponentially with $L$ and $K_\ell$, where $L$ denotes the number of functions and $K_\ell$ denotes the number of data points assigned to each function. Motivated by this scaling, we propose a construction that uses Berrut basis functions as group-selection masks. We derive an approximation-error bound for the proposed construction and establish that the first- and second-order derivatives of the proposed masks at the encoding points scale polynomially with $L$ and $K_\ell$, as $O(L^2K_\ell^2)$ and $O(L^4K_\ell^4)$, respectively. Numerical experiments across non-polynomial functions show that the proposed construction achieves lower relative errors in the tested settings.
\end{abstract}

\vspace{-0.2cm}
\section{Introduction}
\label{sec:intro}

Distributed computing enables large-scale parallel processing but faces challenges from slow workers, commonly referred to as stragglers, Byzantine workers returning incorrect results, and privacy leakage due to curious workers~\cite{a1}. Coded computing addresses these challenges through computational redundancy during task distribution~\cite{c1,w1,w3,f1,a11,a8,a2,SBACC,a9,a10,a12}. Lagrange coded computing (LCC)~\cite{w2} enables resilient polynomial computation in the presence of such workers, while analog Lagrange coded computing (ALCC)~\cite{b5} extends it to the analog domain using floating-point arithmetic. However, many applications, particularly in machine learning, involve non-polynomial functions that these approaches do not directly support~\cite{tandon2017gradient,lee2018speeding}. Berrut approximated coded computing (BACC)~\cite{b6} uses Berrut rational interpolation to approximately compute arbitrary functions despite stragglers. All these frameworks are designed to compute a single function over a given dataset.

In the field of multi-task learning (MTL), multiple tasks are executed, often by exploiting shared information across tasks~\cite{caruana1997multitask,zhang2022survey}. When these tasks are distributed across a common pool of workers, coded computing can leverage computational redundancy to enable resilient evaluation of multiple target functions over different data subsets. Specifically, given $K$ data points $\{\mathbf X_j\}_{j\in[K]}$ and $L$ target functions $\{\psi_\ell\}_{\ell\in[L]}$, the aim is to compute $\psi_\ell(\mathbf X_j)$ for every $j\in G_\ell$, where $G_\ell\subseteq[K]$ denotes the data points assigned to $\psi_\ell$ and $K_\ell=|G_\ell|$. Computing each function separately requires multiple communication rounds, whereas partitioning workers among functions reduces the redundancy available to each function. Thus, computing these evaluations in a single round can reduce communication overhead while retaining redundancy against stragglers. One such method was proposed in~\cite{main}, which combines multiple polynomial functions into a joint function using product-polynomial masks within LCC, enabling exact recovery of multiple function evaluations over finite fields. However, this approach is designed for polynomial computation over finite fields and does not directly support arbitrary non-polynomial functions in the analog domain. To the best of our knowledge, single-round computation of multiple arbitrary non-polynomial functions in the analog domain remains unexplored. 


To address this research gap, we consider the problem of computing multiple arbitrary non-polynomial functions using BACC and develop two single-round constructions in the presence of stragglers. Our contributions are as follows. (i) We develop a BACC counterpart to the single-round LCC approach~\cite{main} using product-polynomial masks and derive an approximation-error bound. We observe that the derivatives of the product-polynomial masks exhibit exponential scaling with $L$ and $K_\ell$, leading to increasingly large error bounds as these parameters grow (see Sections~\ref{subsec:product_approach} and~\ref{subsec:th_product}). (ii) Motivated by this scaling, we propose a construction using Berrut basis functions as group-selection masks and derive an approximation-error bound in terms of the derivatives of the joint function. We also establish that the first- and second-order derivatives of the proposed masks at the encoding points scale polynomially with $L$ and $K_\ell$, as $O(L^2K_\ell^2)$ and $O(L^4K_\ell^4)$, respectively (see Sections~\ref{subsec:proposed_framework} and~\ref{subsec:th:proposed}). This polynomial scaling results in a less severe dependence of the approximation-error bound on $L$ and $K_\ell$ than the exponential factor in the product-polynomial bound. (iii) Finally, we experimentally compare both constructions for multiple non-polynomial functions by varying $L$ and $K_\ell$. The proposed construction achieves lower relative errors than the existing approach in the tested settings (see Section~\ref{sec:experiments}), consistent with the theoretical observation.

\section{System Model and Problem Formulation}
\label{sec:problem_formulation}
We consider the problem of distributed computation of $L>1$ possibly
distinct functions $\psi_\ell(\cdot)$, $\ell\in[L]$, over a common
dataset $\mathcal{X}$. The dataset comprises $K$ data blocks, and
each function is applied to a specified subset of these blocks.
The target functions $\psi_\ell$ may be arbitrary, including
non-polynomial functions.Our setup consists of a master and $N$ workers indexed by
$\mathcal{W}=\{\mathcal{W}_0,\mathcal{W}_1,\ldots,\mathcal{W}_{N-1}\}$. Each worker communicates exclusively with the master through a
dedicated link, with no communication permitted among workers. The dataset $\mathcal{X}=\{\mathbf{X}_i\}_{i\in[K]}$, comprising $K$
data blocks, is stored at the master, where
$[K]\triangleq\{1,2,\ldots,K\}$ and
$\mathbf{X}_i\in\mathbb{R}^{m\times n}$ for every $i\in[K]$.
All workers are assumed to have prior knowledge of the functions
to be evaluated. We assume that $s$ workers, where $1\leq s<N-3$, may fail to return their computation results within
the prescribed time. The master aims to compute the $L>1$ desired
functions, each applied to a specified subset of the data blocks,
while ensuring bounded approximation error despite the presence of stragglers. Let $G_\ell\subseteq[K]$ denote the set of indices of data blocks
assigned to $\psi_\ell$, such that $G_1,\ldots,G_L$ form a
partition of $[K]$, i.e.,
$\bigcup_{\ell=1}^{L}G_\ell=[K]$ and
$G_\ell\cap G_r=\emptyset$ for $\ell\neq r$. Let
$K_\ell\triangleq |G_\ell|$. The desired computations are therefore
$\mathcal{F}\triangleq
\{\psi_\ell(\mathbf{X}_i):i\in G_\ell,\ell\in[L]\}$
\footnote{We assume that all target functions $\psi_\ell$,
$\ell\in[L]$, share a common input space and a common output space,
ensuring that they can be evaluated on the encoded data and
combined into a single worker computation function.} Since the target functions may be polynomial or non-polynomial, we
adopt the Berrut Approximated Coded Computing (BACC)
framework~\cite{b6}, which enables coded computation of arbitrary
functions through Berrut rational interpolation. To this end, for distinct, ordered points $\{x_j\}_{j\in[n]}$, such that $x_1<x_2<\cdots<x_n$, we define the Berrut basis
functions as

\begin{scriptsize}
\begin{equation}
\label{eq:phi(z)}
\phi_{x_j}(z)\triangleq
\frac{\dfrac{(-1)^j}{z-x_j}}
{\displaystyle\sum_{q=1}^{n}\dfrac{(-1)^q}{z-x_q}}.
\end{equation}
\end{scriptsize}

\noindent These basis functions satisfy the cardinal property
$\phi_{x_j}(x_i)=\delta_{ij}$, where $\delta_{ij}$ denotes the
Kronecker delta. We use these basis functions to construct a Berrut interpolant for
encoding the dataset $\mathcal{X}$ at the encoding points
$\{x_j\}_{j\in[n]}$ and subsequently approximate the desired
computations
$\{\psi_\ell(\mathbf{X}_i):i\in G_\ell,\ell\in[L]\}$
from the available worker responses.

\subsubsection{Encoding}
Let $\mathcal{A}\triangleq\{\alpha_i\}_{i\in[K]}$ denote the set of
distinct encoding points. We choose the encoding points as the Chebyshev nodes of
the first kind, i.e.,
$\alpha_i=\cos\frac{(2i-1)\pi}{2K}$ for $i\in[K]$. Using the Berrut basis functions in
\eqref{eq:phi(z)} associated with these points, the master constructs
the common encoding function

\begin{equation}
\label{eq:u(z)}
\mathbf{u}(z)\triangleq
\sum_{i=1}^{K}\phi_{\alpha_i}(z)\mathbf{X}_i,
\end{equation}
which satisfies $\mathbf{u}(\alpha_i)=\mathbf{X}_i$ for every
$i\in[K]$.

\subsubsection{Distribution of Shares and Joint Computation}
The master selects $N$ distinct evaluation points
$\mathcal{B}\triangleq\{\beta_i\}_{i\in[N]}$ for the $N$ workers,
chosen as the Chebyshev nodes of the second kind, i.e.,
$\beta_i=\cos\frac{(i-1)\pi}{N-1}$ for $i\in[N]$.
The master distributes the encoded share $\mathbf{u}(\beta_i)$ to
worker $i$ and assigns it the task of computing a joint function
$h(z)$, formed by combining the $L$ target-function evaluations
$\psi_\ell(\mathbf{u}(z))$, $\ell\in[L]$. Each worker evaluates
$h(\beta_i)$ using its received encoded share and returns the result to the master, enabling the desired outputs to be approximated in a single communication round.

\subsubsection{Reconstruction}
Suppose that $N_1=N-s$ workers respond, where $s$ workers are
stragglers, and let
$\mathcal{S}\triangleq\{m_1,\ldots,m_{N_1}\}\subseteq[N]$, with
$m_1<\cdots<m_{N_1}$, denote the set of indices of responding
workers. For each $q\in[N_1]$, let
$\widetilde{\beta}_q=\beta_{m_q}$ denote the corresponding
evaluation point. Using the available responses
$\mathbf{y}_{m_q}=h(\widetilde{\beta}_q)$, the master approximates
$h(z)$ using Berrut interpolation as
\begin{equation}
\label{eq:r_berrut}
r(z)\triangleq
\sum_{q=1}^{N_1}
\phi_{\widetilde{\beta}_q}(z)\mathbf{y}_{m_q}
\approx h(z).
\end{equation}
Since $\mathbf{u}(\alpha_i)=\mathbf{X}_i$ and
$h(\alpha_i)=\psi_\ell(\mathbf{X}_i)$ for $i\in G_\ell$, the master
evaluates $r(z)$ at the encoding points $\alpha_i$, $i\in[K]$, to
obtain $r(\alpha_i)\approx\psi_\ell(\mathbf{X}_i)$ for
$i\in G_\ell$, $\ell\in[L]$. Thus, the approximation accuracy of
the desired computations depends on the choice of the joint
function $h(z)$. This motivates the central question of this work:
\emph{Can we construct a joint function $h(z)$ that combines all
$L$ target computations and enables approximate recovery of the
desired outputs $\{\psi_\ell(\mathbf{X}_i)\!:\!i\in\! G_\ell,\ell\in[L]\}$ from a subset of worker responses in a single
communication round?}

\section{Distributed Computing Schemes for Multiple Arbitrary Functions in BACC}
\label{sec:multiple function ides}
In this section, we present approaches for computing multiple arbitrary functions, including polynomial and non-polynomial functions, using BACC. We first discuss two naive extensions, then adapt the multi-function construction of \cite{main}, originally developed for LCC, to the BACC framework, and finally present our proposed single-round construction.
\vspace{-0.1cm}
\subsection{Naive Approaches for Multiple-Function Computation}
\label{subsec:naive_approaches}
We first consider two straightforward extensions of BACC for computing the desired set of multiple-function outputs $\mathcal{F}$.
\vspace{0.1cm}
\subsubsection{\textbf{Separate \texorpdfstring{$L$}{L}-Round BACC}}
A straightforward approach is to apply BACC independently to each
target function in a separate communication round. In round
$\ell\in[L]$, the master considers the data blocks indexed by
$G_\ell$ and constructs an encoding function $\mathbf{u}_\ell(z)$
as in \eqref{eq:u(z)}, satisfying
$\mathbf{u}_\ell(\alpha_{i,\ell})=\mathbf{X}_i$ for every
$i\!\in \!G_\ell$. Each worker $i\!\in\!\![N]$ evaluates
$h_\ell(\beta_i)\!=\!\psi_\ell(\mathbf{u}_\ell(\beta_i))$.
Given the responding worker set $\mathcal{S}_\ell\subseteq[N]$
in round $\ell$, the master applies Berrut interpolation as in
\eqref{eq:r_berrut} to obtain $r_\ell(z)$ and approximates the outputs as
$r_\ell(\alpha_{i,\ell})\approx\psi_\ell(\mathbf{X}_i)$ for
$i\in G_\ell$.


\subsubsection{\textbf{Worker-Partitioned BACC}}
The master partitions the worker set $\mathcal{W}$ into $L$ disjoint
groups, $\{\mathcal{W}_\ell\}_{\ell=1}^{L}$, such that
$\bigcup_{\ell=1}^{L}\mathcal{W}_\ell=\mathcal{W}$ and
$\mathcal{W}_\ell\cap\mathcal{W}_r=\emptyset$ for $\ell\neq r$,
assigning approximately $N/L$ workers to each function $\psi_\ell$. For each
$\ell\in[L]$, it constructs
$\mathbf{u}_\ell(z)$ from the data blocks indexed by $G_\ell$ as in \eqref{eq:u(z)}.
Each worker $i\in\mathcal{W}_\ell$ evaluates
$h_\ell(\beta_{i,\ell})
=\psi_\ell(\mathbf{u}_\ell(\beta_{i,\ell}))$, where
$\beta_{i,\ell}$ is its assigned evaluation point. Using the
responses from $\mathcal{S}_\ell\subseteq\mathcal{W}_\ell$,
the master constructs the Berrut interpolant $r_\ell(z)$ as in
\eqref{eq:r_berrut} and recovers the desired outputs approximately
as $r_\ell(\alpha_{i,\ell})\approx\psi_\ell(\mathbf{X}_i)$
for $i\!\in\! G_\ell$. Although this approach computes all $L$ functions in parallel
within a single communication round, partitioning the workers
reduces the number of available workers for each function.
In the worst case, all workers assigned to a function may be
stragglers, preventing the master from approximating the outputs
associated with that group.

The above two naive approaches either require $L$ sequential communication rounds or partition the available workers among the target functions, potentially limiting straggler resilience. To avoid these limitations, we seek a joint worker function $h(z)$ that incorporates all $L$ computations and enables all $N$ workers to contribute to approximation
in a single communication round. To this end, we develop a BACC counterpart to the single-round multiple-function construction in~\cite{main}, originally proposed for LCC. We refer to this scheme as the \emph{Product-based Construction} and use it as a baseline.

\subsection{Product-Based Construction for Multiple-Function BACC}
\label{subsec:product_approach}
The product-based construction combines the $L$ function computations into a single worker function using polynomial masks. For each $\ell\in[L]$, the mask $P_\ell(z)$ is designed to vanish at the
encoding points associated with all function groups other than $G_\ell$, thereby isolating the $\ell$-th function at its assigned
encoding points. Specifically, define $P_\ell(z)\triangleq\prod_{j\notin G_\ell}(z-\alpha_j),$
which satisfies $P_\ell(\alpha_i)=0$ for every $i\notin G_\ell$. The joint worker function is then defined as 
\begin{equation}
\label{eq:product}
h_{\mathrm P}(z)\triangleq
\sum_{\ell=1}^{L}\psi_\ell(\mathbf u(z))P_\ell(z).
\end{equation}
Each worker $i\in[N]$ evaluates $h_{\mathrm P}(\beta_i)$ and returns
the result to the master. Suppose that $N_1$ workers respond. Let
$m_1<\cdots<m_{N_1}$ denote the ordered elements of $\mathcal S$,
and define the corresponding evaluation points as
$\widetilde{\beta}_q=\beta_{m_q}$ for $q\in[N_1]$.
Using the received evaluations $h_{\mathrm P}(\widetilde{\beta}_q)$,
the master constructs the Berrut interpolant $r_{\mathrm P}(z)$, as in \eqref{eq:r_berrut}, such that $r_{\mathrm P}(z)\approx h_{\mathrm P}(z)$. The master then evaluates $r_{\mathrm P}(z)$ at the encoding points
$\{\alpha_i\}_{i=1}^K$ to recover the desired function values approximately. Consider an encoding point $\alpha_i$ with $i\in G_\ell$. For any
$r\neq\ell$, we have $i\notin G_r$, so $P_r(z)$ contains the factor
$(z-\alpha_i)$, yielding $P_r(\alpha_i)=0$. Thus, all terms in
$h_{\mathrm P}(\alpha_i)$ except the $\ell$-th term vanish, giving
the exact identity
$h_{\mathrm P}(\alpha_i)
=\psi_\ell\bigl(\mathbf u(\alpha_i)\bigr)P_\ell(\alpha_i).$ Since $\mathbf u(\alpha_i)=\mathbf X_i$ for $i\in G_\ell$, we obtain
\[r_{\mathrm P}(\alpha_i)
\approx P_\ell(\alpha_i)\psi_\ell(\mathbf X_i).\]
Consequently, the desired function value is approximated by normalizing the reconstructed value by $P_\ell(\alpha_i)$ as

\begin{scriptsize}
\begin{equation}
\label{eq:normalisation}
\psi_\ell(\mathbf X_i)
\approx\frac{r_{\mathrm P}(\alpha_i)}{P_\ell(\alpha_i)}
=\frac{r_{\mathrm P}(\alpha_i)}
{\displaystyle\prod_{j\notin G_\ell}(\alpha_i-\alpha_j)},
\qquad i\in G_\ell,\ \ell\in[L].
\end{equation}
\end{scriptsize}

Thus, the product-based construction requires normalization by
$P_\ell(\alpha_i)$ after Berrut reconstruction, as shown in
\eqref{eq:normalisation}. The degree of $P_\ell(z)$ is
$K\!-\!K_\ell$. When all $L$ groups have
equal size, i.e., $K_\ell\!=\!K/L$, the degree reduces to $K(1-1/L)$. 

\begin{table}[t]
\centering
\caption{Comparison of multiple-function BACC schemes.}
\label{tab:complexity}
\footnotesize
\setlength{\tabcolsep}{1.6pt}
\renewcommand{\arraystretch}{1.11}
\begin{tabular}{|l|c|c|c|c|}
\hline
Scheme & Rounds & \shortstack{Workers per\\function} &
\shortstack{Function evals.\\per worker} &
\shortstack{Post-reconstruction\\normalization} \\
\hline
Separate BACC
& $L$ & $N$ & $1$ & No \\
Worker-partitioned
& $1$ & $\approx N/L$ & $1$ & No \\
Product-based
& $1$ & $N$ & $L$ & Yes \\
Proposed
& $1$ & $N$ & $L$ & No \\
\hline
\end{tabular}
\vspace{-0.6cm}
\end{table}

Although the Product-based Construction enables single-round
computation, it requires normalizing the reconstructed values by
$P_\ell(\alpha_i)$, as shown in \eqref{eq:normalisation}. Since
reconstruction is approximate, division by a small
$\lvert P_\ell(\alpha_i)\rvert$ can amplify reconstruction errors.
This limitation motivates an alternative construction that enables
single-round computation without post-reconstruction normalization.
Specifically, we seek a collection of functions $M_\ell(z)$
satisfying the group-selection property:

\begin{scriptsize}
\[
\begin{aligned}
M_\ell(\alpha_i)&=1, && i\in G_\ell,\\
M_r(\alpha_i)&=0, && i\in G_\ell,\quad r\neq\ell.
\end{aligned}
\]
\end{scriptsize}
This property ensures that, at each encoding point, only the
function assigned to its group contributes to the joint worker
function. In the following subsection, we construct such functions
using the global Berrut basis and present our proposed method for
computing multiple functions in BACC.
\subsection{Proposed Construction for Multiple-Function BACC}
\label{subsec:proposed_framework}
We construct a single-round BACC scheme by combining the $L$ function
computations into a common coded function. Using the common Berrut
encoding function $\mathbf{u}(z)$ in \eqref{eq:u(z)}, define
$F_\ell(z)\triangleq\psi_\ell(\mathbf{u}(z))$ for $\ell\in[L]$.
Since $\mathbf{u}(\alpha_i)=\mathbf{X}_i$, we have
$F_\ell(\alpha_i)=\psi_\ell(\mathbf{X}_i)$ for every $i\in[K]$. Thus, the key idea is to construct functions that select the
appropriate computation at each encoding point $\alpha_i$. Using the global Berrut basis
functions $\{\phi_{\alpha_i}(z)\}_{i\in[K]}$ defined in
\eqref{eq:phi(z)}, we define
\begin{equation}
\label{eq:group_selection}
M_\ell(z)\triangleq\sum_{i\in G_\ell}\phi_{\alpha_i}(z),
\quad \ell\in[L].
\end{equation}
\noindent By the cardinal property of the Berrut basis, $M_\ell(\alpha_i)=1$ for $i\!\in\! G_\ell$ and $M_r(\alpha_i)=0$ for $r\neq\ell$. Moreover, since
$G_1,\ldots,G_L$ partition $[K]$, the functions satisfy
$\sum_{\ell=1}^{L}M_\ell(z)=1$ wherever the Berrut basis is
defined. We then combine the $L$ computations into the joint computation function as

\begin{scriptsize}
\begin{equation}
\label{eq:proposed}
h_B(z)\triangleq
\sum_{\ell=1}^{L}F_\ell(z)M_\ell(z)
=\sum_{\ell=1}^{L}\psi_\ell(\mathbf{u}(z))M_\ell(z).
\end{equation}
\end{scriptsize}

Worker $i$ evaluates this function at $\beta_i$ and returns
\begin{scriptsize}
\begin{equation}
\label{eq:proposed_eval}
h_B(\beta_i)=
\sum_{\ell=1}^{L}
\psi_\ell(\mathbf{u}(\beta_i))M_\ell(\beta_i).
\end{equation}
\end{scriptsize}

 Suppose that $N_1<N$ workers respond, with ordered indices
$m_1<\cdots<m_{N_1}$ from $\mathcal S$ and corresponding evaluation
points $\widetilde{\beta}_q=\beta_{m_q}$, $q\in[N_1]$. Using the received evaluations $h_B(\widetilde{\beta}_q)$, the master constructs the Berrut interpolant $r_B(z)$ according to \eqref{eq:r_berrut}, yielding $r_B(z)\approx h_B(z)$. For any $j\in G_\ell$, the group-selection property gives $h_B(\alpha_j)=\psi_\ell(\mathbf X_j)$. Therefore, evaluating the interpolant at $\alpha_j$ yields $r_B(\alpha_j)\approx\psi_\ell(\mathbf X_j)$.

\begin{remark}
For $j\!\in\! G_\ell$, the proposed construction yields
$r_B(\alpha_j)\!\approx\!\psi_\ell(\mathbf{X}_j)$. Thus, unlike the
Product-based Construction in \eqref{eq:normalisation}, the proposed
scheme requires no additional normalization to recover the desired
computations $\psi_\ell(\mathbf{X}_j)$ for $l\in[L]$.
\end{remark}

\section{Theoretical Analysis}
\label{sec:theoretical analysis}

In this section, we analyze the approximation errors of the
Product-based and the proposed multiple-function
construction for BACC in the presence of stragglers. For the analysis, we assume that all $L$ groups have equal size, such that $K_\ell=K/L$ is an integer for every $\ell\in[L]$. Specifically, we bound
$\|r_{\mathrm P}-h_{\mathrm P}\|_\infty$ and
$\|r_B-h_B\|_\infty$ and characterize their dependence on $L$,
$K_\ell$, and the number of stragglers $s$. Here,
$\|f\|_\infty\triangleq\sup_{z\in[-1,1]}|f(z)|$ denotes the
supremum norm. To this end, for single-function computation, i.e., $L=1$, Theorem~9 of \cite{b6} bounds the approximation error of BACC in terms of the first and second derivatives of the underlying function, with additional dependence on the mesh ratio $\lambda$ and spacing parameter $h$ of the evaluation points $\{\beta_i\}_{i=1}^{N}$, as presented in Theorem~5 of \cite{b6}. We extend this analysis to multiple-function computation, i.e., $L>1$, by expressing the joint function for both constructions as $h(z)=\sum_{\ell=1}^{L}F_\ell(z)Q_\ell(z)$, where $F_\ell(z)=\psi_\ell(\mathbf u(z))$. For the Product-based
construction, $Q_\ell(z)=P_\ell(z)$ and $h(z)=h_{\mathrm P}(z)$, and
for the proposed construction, $Q_\ell(z)=M_\ell(z)$ and
$h(z)=h_B(z)$, as defined in \eqref{eq:product} and
\eqref{eq:proposed}, respectively. 
Therefore, extending the bound from $L=1$ to $L>1$ requires
bounding the first and second derivatives of $h(z)$, given by

\begin{scriptsize}
\begin{equation}
\label{eq:generic_first_derivative}
h'(z)=\sum_{\ell=1}^{L}
\left[F_\ell'(z)Q_\ell(z)
+F_\ell(z)Q_\ell'(z)\right],
\end{equation}
\vspace{-0.3cm}
\begin{equation}
\label{eq:generic_second_derivative}
h''(z)=\sum_{\ell=1}^{L}
\left[
F_\ell''(z)Q_\ell(z)
+2F_\ell'(z)Q_\ell'(z)
+F_\ell(z)Q_\ell''(z)
\right].
\end{equation}
\end{scriptsize}

Thus, for $L>1$, the error bound depends not only on the functions $F_\ell(z)$ and their derivatives but also on the magnitudes and derivatives of $Q_\ell(z)$, as shown in \eqref{eq:generic_first_derivative} and \eqref{eq:generic_second_derivative}. Since $Q_\ell(z)$ differs between the two constructions, we analyze them separately to derive error bounds in terms of $L$, $K_\ell$, and $s$.
\vspace{-0.1cm}
\subsection{Approximation Error of the Product-Based Construction}
\label{subsec:th_product}
For the Product-based construction, substituting
$Q_\ell(z)=P_\ell(z)$ into \eqref{eq:generic_first_derivative}
and \eqref{eq:generic_second_derivative} gives the corresponding
derivative expressions for $h_{\mathrm P}(z)$. Combining these
expressions with Theorem~9 of~\cite{b6} gives the following
error bound.
\begin{theorem}
\label{thm:product_error}
Suppose the $K$ data blocks are partitioned into $L\geq 2$
equal-sized groups, such that $K_\ell=K/L$ is an integer for every
$\ell\in[L]$. Define
$d\triangleq K-K_\ell=(L-1)K_\ell$, and let
$F_\ell(z)\triangleq\psi_\ell(\mathbf u(z))
\in C^2([-1,1])$ for all $\ell\in[L]$.
Consider $N$ workers with $s$ stragglers, where $s<N-3$, and let
$r_{\mathrm P}$ denote the Berrut reconstruction of the
Product-based joint function $h_{\mathrm P}$ as defined in \eqref{eq:product}, and
$R_s=(s+1)(s+3)\pi^2/4$. If $N-s$ is even, then

\vspace{-0.2cm}
\begin{scriptsize}
\begin{equation*}
\label{eq:product_error_even_Ns}
\begin{aligned}
\|r_{\mathrm P}-h_{\mathrm P}\|_\infty
\leq{}&
2(1+R_s)
\sin\left(\frac{(s+1)\pi}{2(N-1)}\right)2^d\\
&\times\sum_{\ell=1}^{L}
\left[\|F_\ell''\|_\infty+2d^2\|F_\ell'\|_\infty
+\frac{d^2(d^2-1)}{3}\|F_\ell\|_\infty\right].
\end{aligned}
\end{equation*}
\end{scriptsize}
\vspace{-0.2cm}
If $N-s$ is odd, then
\begin{scriptsize}
\begin{equation*}
\label{eq:product_error_odd_Ns}
\begin{aligned}
\|r_{\mathrm P}-h_{\mathrm P}\|_\infty
\leq{}
2(1+R_s)
\sin\left(\frac{(s+1)\pi}{2(N-1)}\right)2^d\times\sum_{\ell=1}^{L}
\Bigg[\|F_\ell''\|_\infty+(2d^2+1)\\\|F_\ell'\|_\infty
+\left(d^2+\frac{d^2(d^2-1)}{3}\right)
\|F_\ell\|_\infty
\Bigg].
\end{aligned}
\end{equation*}
\end{scriptsize}
\noindent Furthermore, suppose there exists a constant $C>0$, independent of
$K$ and $L$, such that $\max_{\ell\in[L],\,q\in\{0,1,2\}}
\|F_\ell^{(q)}\|_\infty\leq C,$
where $F_\ell^{(0)}\triangleq F_\ell$. Under this additional
assumption, for fixed $s$, the upper bounds in the theorem scale as
\begin{scriptsize}
\begin{equation}
\label{eq:order_product}
\|r_{\mathrm P}-h_{\mathrm P}\|_\infty
=O\left(L(L-1)^4K_\ell^42^{(L-1)K_\ell}\sin\left(\frac{(s+1)\pi}{2(N-1)}\right)
\right).
\end{equation}
\end{scriptsize}
\end{theorem}


\begin{corollary}
\label{cor:product_output_error}
Under the conditions of Theorem~\ref{thm:product_error}, for every
$i\in G_\ell$, the recovered value satisfies

\begin{scriptsize}
$$\left|\widehat{\psi_\ell(\mathbf X_i)}-\psi_\ell(\mathbf X_i)
\right|\leq\frac{\|r_{\mathrm P}-h_{\mathrm P}\|_\infty}
{\displaystyle\prod_{j\notin G_\ell}|\alpha_i-\alpha_j|}.$$
\end{scriptsize}

\end{corollary}

\begin{proof}
For $i\in G_\ell$, the Product-based construction gives
$h_{\mathrm P}(\alpha_i)=P_\ell(\alpha_i)\psi_\ell(\mathbf X_i).$
Hence, $\widehat{\psi_\ell(\mathbf X_i)}=\!\!\frac{r_{\mathrm P}(\alpha_i)}{P_\ell(\alpha_i)}.$ Therefore,
\begin{scriptsize}
$$\left|\widehat{\psi_\ell(\mathbf X_i)}-\psi_\ell(\mathbf X_i)
\right|=\frac{|r_{\mathrm P}(\alpha_i)-h_{\mathrm P}(\alpha_i)|
}{|P_\ell(\alpha_i)|}\leq\frac{\|r_{\mathrm P}-h_{\mathrm P}\|_\infty
}{|P_\ell(\alpha_i)|}.$$
\end{scriptsize}
Since $P_\ell(\alpha_i)=\prod_{j\notin G_\ell}(\alpha_i-\alpha_j),$
the result follows.
\end{proof}

\begin{figure*}[t]
    \centering
    \begin{subfigure}[b]{0.245\textwidth}
        \centering
        \includegraphics[width=\linewidth]{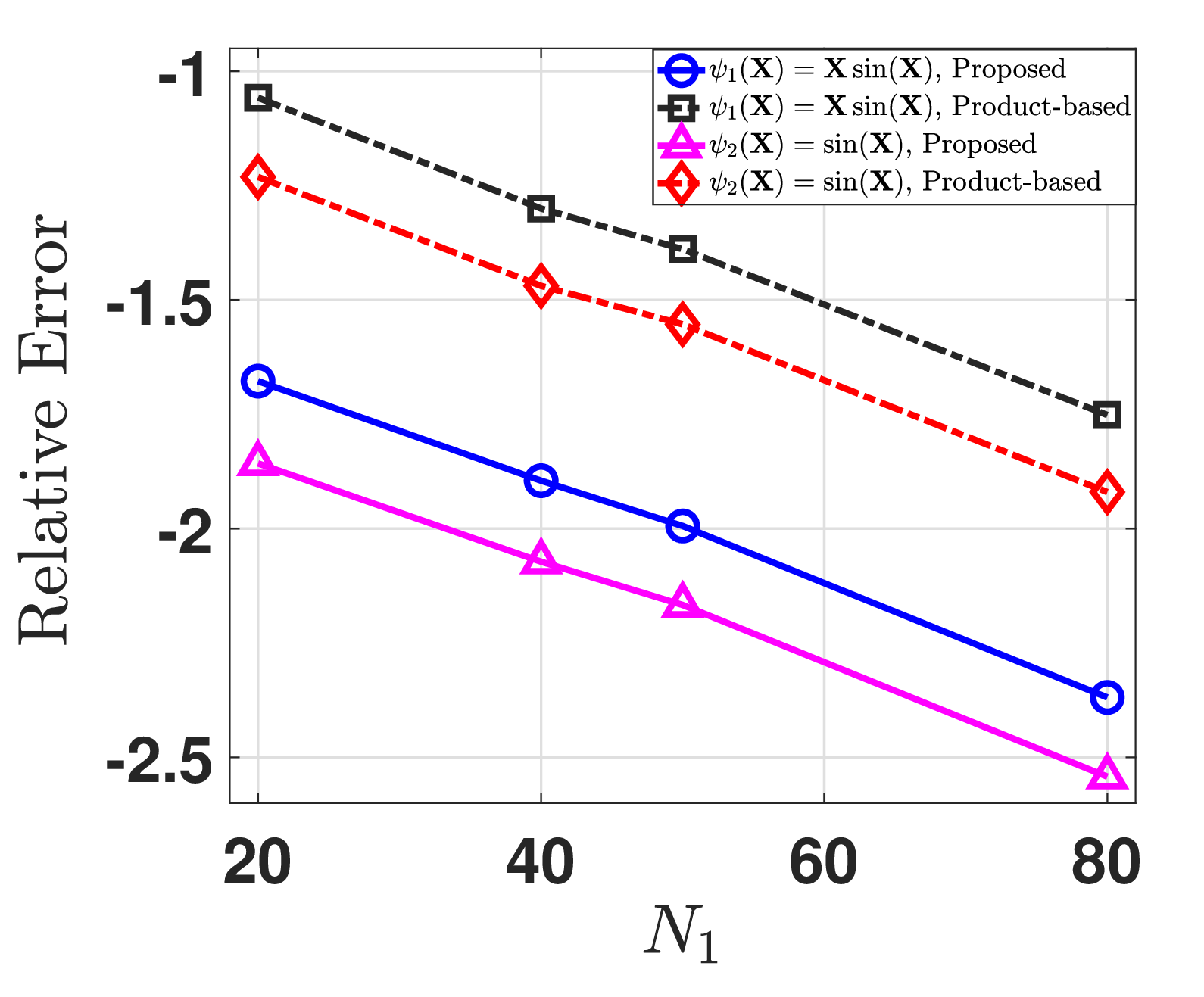}
        \subcaption{}
    \end{subfigure}
    \begin{subfigure}[b]{0.245\textwidth}
        \centering
        \includegraphics[width=\linewidth]{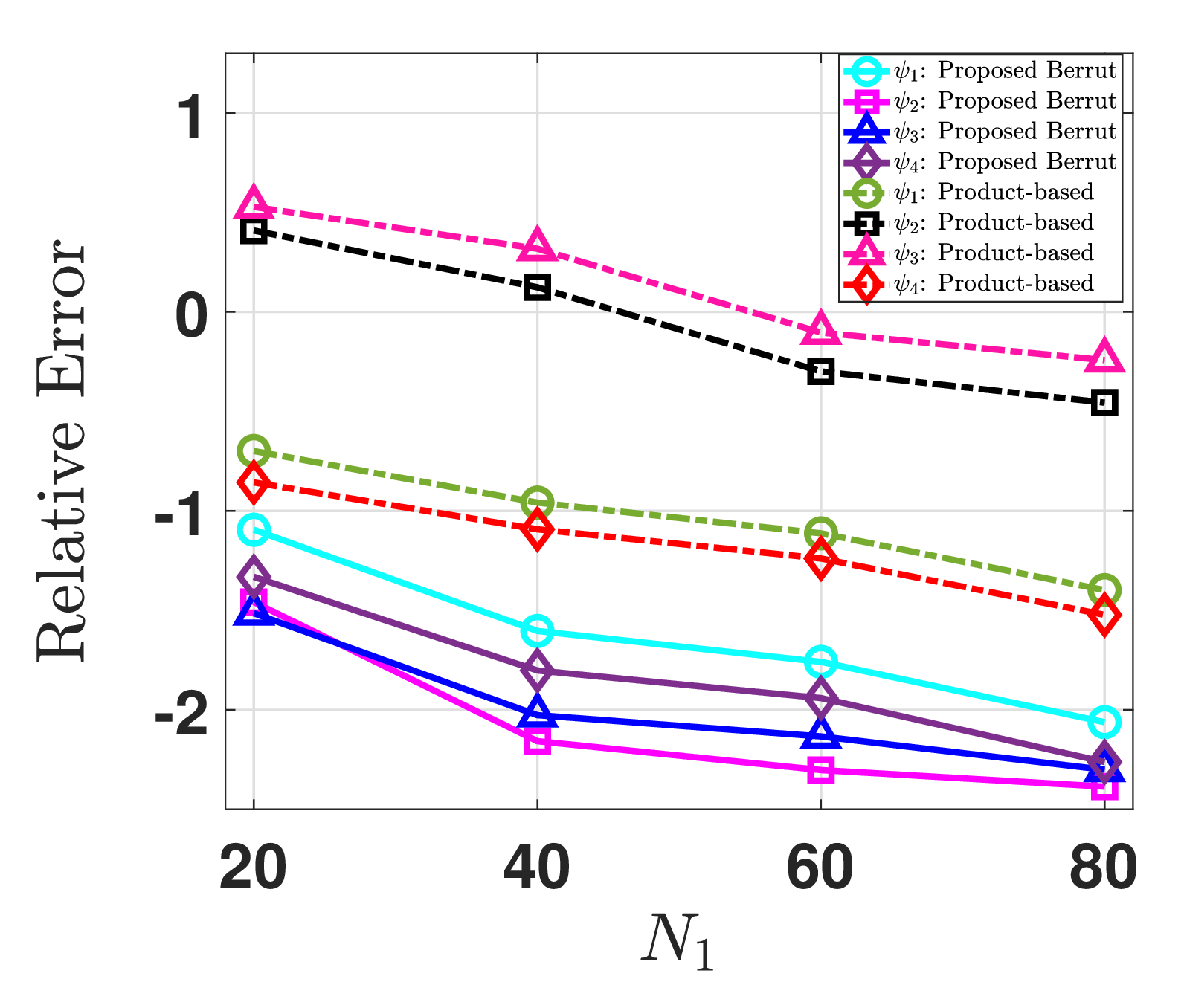}
        \subcaption{}
    \end{subfigure}
    \begin{subfigure}[b]{0.245\textwidth}
        \centering
        \includegraphics[width=\linewidth]{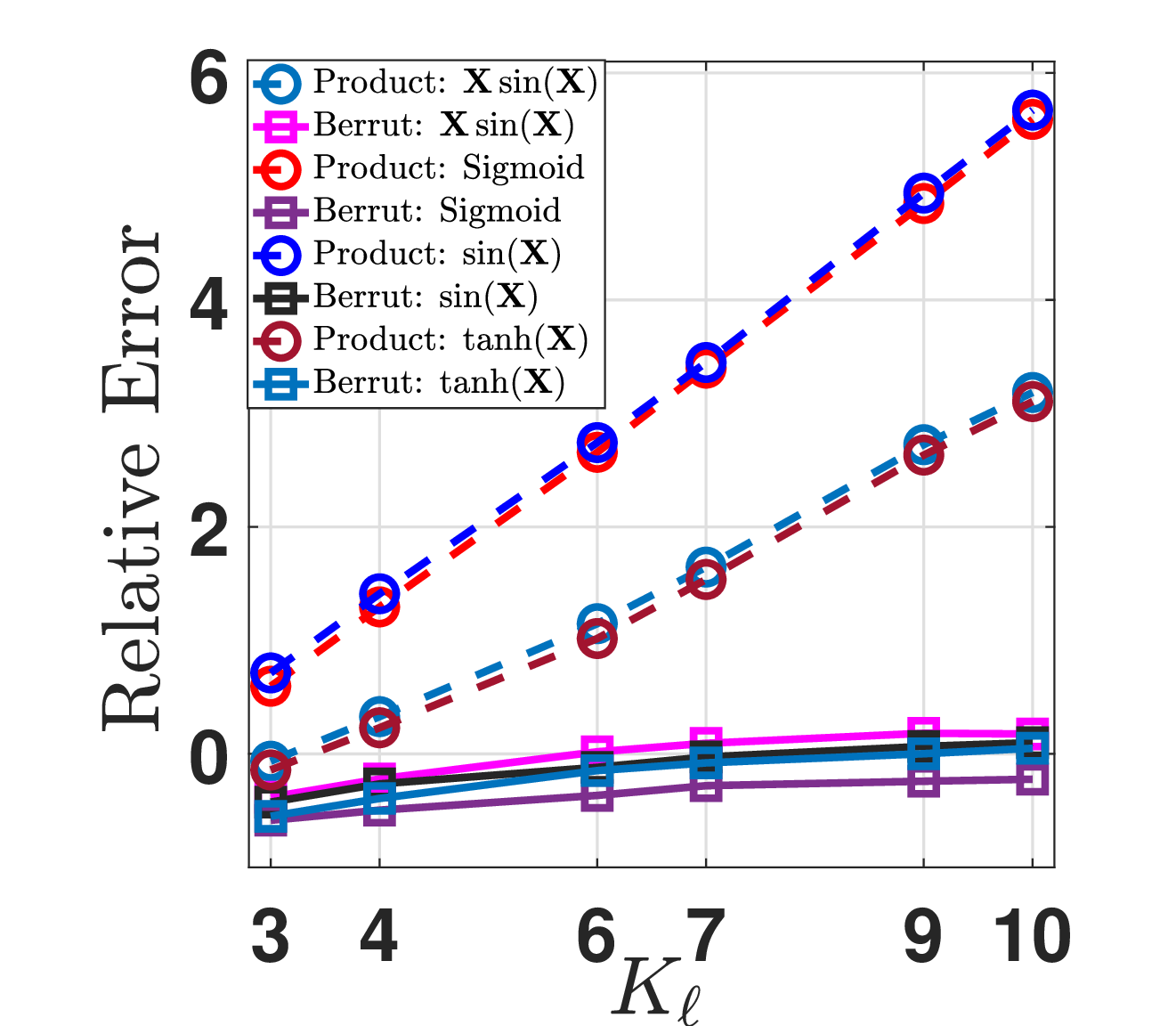}
        \subcaption{}
    \end{subfigure}
    \begin{subfigure}[b]{0.245\textwidth}
        \centering
        \includegraphics[width=\linewidth]{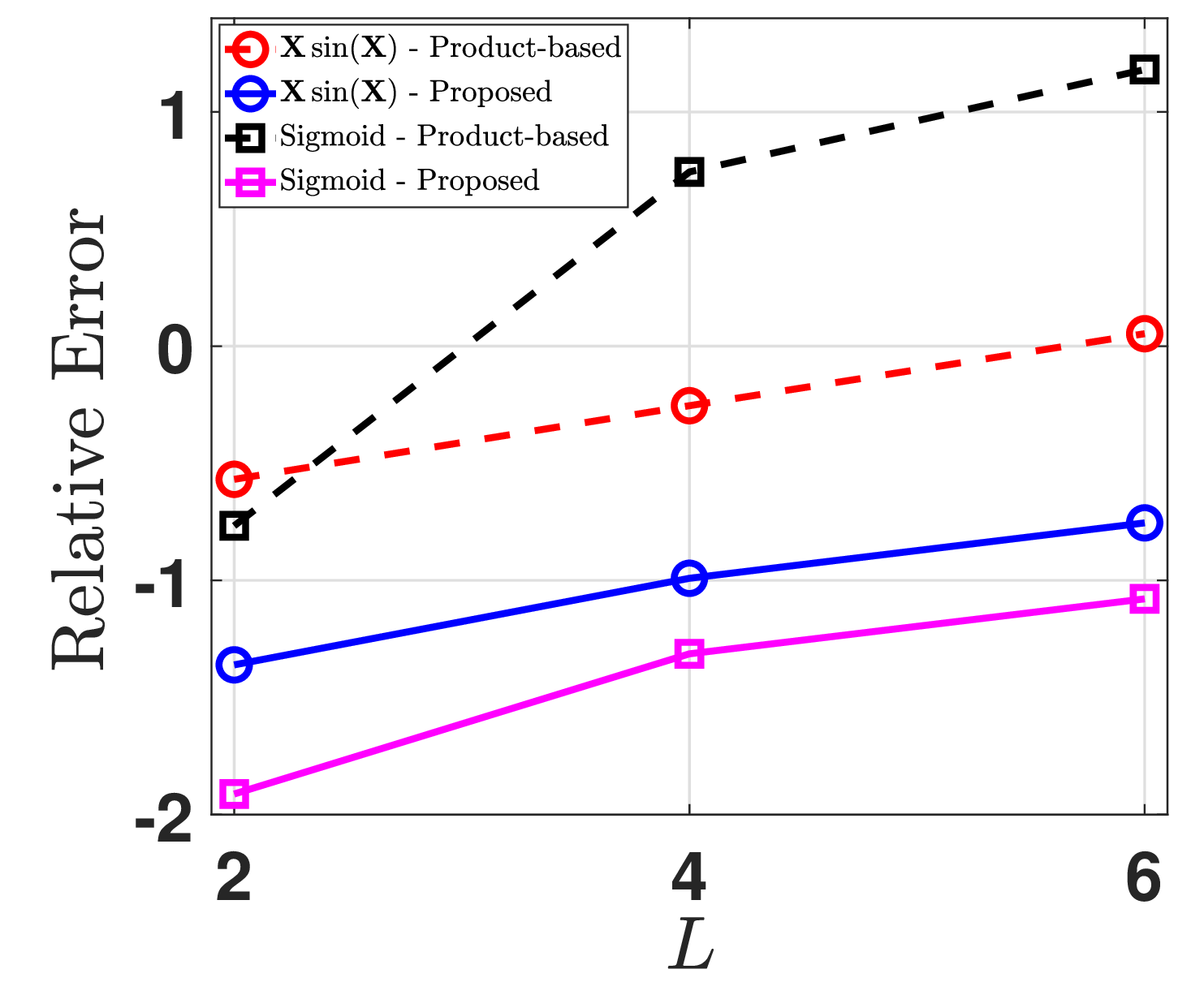}
        \subcaption{}
    \end{subfigure}
    \vspace{-0.65cm}
    \caption{Average relative error in $\log_{10}$ scale:
    (a) versus $N_1$ for $\psi_1(\mathbf{X})=\mathbf{X}\sin(\mathbf{X})$
    and $\psi_3(\mathbf{X})=\sin(\mathbf{X})$,
    (b) versus $N_1$ for all functions $\psi_1(\mathbf{X})=\mathbf{X}\sin(\mathbf{X})$,
$\psi_2(\mathbf{X})=(1+e^{-\mathbf{X}})^{-1}$,
$\psi_3(\mathbf{X})=\sin(\mathbf{X})$, and
$\psi_4(\mathbf{X})=\tanh(\mathbf{X})$,
    (c) versus $K_\ell$ for $\psi_1(\mathbf{X})=\mathbf{X}\sin(\mathbf{X})$,
    with $L=4$ and $N_1=40$, and
    (d) versus $L$, with fixed $K_\ell=4$ and $N_1=40$,
    for $\psi_1(\mathbf{X})=\mathbf{X}\sin(\mathbf{X})$.}
    \label{fig:experimental_results}
        \vspace{-0.4cm}
\end{figure*}
\begin{corollary}
\label{cor:product_K_L}
The upper bound on $\|r_{\mathrm P}-h_{\mathrm P}\|_\infty$
in Theorem~\ref{thm:product_error} contains the exponential factor
$2^d=2^{(L-1)K_\ell}$, where $d=(L-1)K_\ell$ is the degree of
$P_\ell(z)$. For fixed $L>1$, this factor is exponential in
$K_\ell$. For fixed $K_\ell$, the bound depends on $L$ through
both the exponential factor and the multiplicative factor
$L(L-1)^4$.
\end{corollary}
\vspace{-0.3cm}
\subsection{Approximation Error of the Proposed Construction}
\label{subsec:th:proposed}
\vspace{-0.1cm}
For the proposed construction, substituting
$Q_\ell(z)=M_\ell(z)$ into \eqref{eq:generic_first_derivative}
and \eqref{eq:generic_second_derivative}, where $M_\ell(z)$ is defined in
\eqref{eq:group_selection}, gives the first and second derivatives of $h_B(z)$. Using the approximation-error
result in Theorem~9 of~\cite{b6}, we obtain the following bound.

\begin{theorem}
\label{thm:global_berrut_error_alpha}
Suppose that the groups have equal size $K_\ell=K/L$, where
$L\geq 2$ and $K/L$ is an integer, and that
$F_\ell(z)\in C^2([-1,1])$ for all $\ell\in[L]$. Consider $N$ workers
with $s$ stragglers, where $s<N-3$, and let $r_{\mathrm B}$ denote
the Berrut reconstruction of $h_B$ as defined in
\eqref{eq:r_berrut}. Define $A_r\triangleq\max_{\ell\in[L]}\|F_\ell^{(r)}\|_\infty$, and $\Lambda_K^{(r)}
\triangleq\max_{z\in[-1,1]}\sum_{i=1}^{K}\left|\phi_{\alpha_i}^{(r)}(z)\right|$ for $r=0,1,2,$, where $F_\ell^{(0)}\triangleq F_\ell$ and
$\|\cdot\|_\infty$ denotes the supremum norm on $[-1,1]$. If $N-s$ is even, then

\begin{scriptsize}
\begin{equation}
\label{eq:global_berrut_error_even}
\|r_{\mathrm B}-h_B\|_\infty
\leq
C_{N,s}
\left[
A_2\Lambda_K^{(0)}
+2A_1\Lambda_K^{(1)}
+A_0\Lambda_K^{(2)}
\right].
\end{equation}
\end{scriptsize}
If $N-s$ is odd, then
\begin{scriptsize}
\begin{equation}
\label{eq:global_berrut_error_odd}
\begin{aligned}
\|r_{\mathrm B}-h_B\|_\infty
\leq C_{N,s}\Big[A_2\Lambda_K^{(0)}+2A_1\Lambda_K^{(1)}+A_0\Lambda_K^{(2)}+A_1\Lambda_K^{(0)}
+A_0\Lambda_K^{(1)}
\Big].
\end{aligned}
\end{equation}
\end{scriptsize}
\noindent where $C_{N,s}=2(1+R_s)\sin\left(\frac{(s+1)\pi}{2(N-1)}\right),  R_s=\frac{(s+1)(s+3)\pi^2}{4}.$
\end{theorem}
Theorem~\ref{thm:global_berrut_error_alpha} bounds the global
approximation error in terms of the derivative measures
$\Lambda_K^{(r)}$, which involve a supremum over the entire
interval $[-1,1]$. Since the desired function values are
recovered at the encoding points $\{\alpha_i\}_{i=1}^K$,
where the proposed construction satisfies
$h_B(\alpha_i)=\psi_\ell(\mathbf X_i)$ for $i\in G_\ell$,
and the reconstructed function approximates $h_B(z)$, as
$r_B(\alpha_i)\approx h_B(\alpha_i)
=\psi_\ell(\mathbf X_i)$ for $ i\in G_\ell.$
Motivated by this relationship, we examine the derivatives of
the Berrut basis functions at the encoding points. For $r=1,2$,
we define the corresponding derivative measures at these points
as $\Lambda_{K,\alpha}^{(r)}$. We first derive bounds for these
measures for arbitrary distinct encoding points and then
establish asymptotic upper bounds for Chebyshev nodes of the
first kind, which are used in BACC.
\begin{lemma}
\label{prop:encoding_point_derivative}
Consider the Berrut basis functions $\phi_{\alpha_i}(z)$,
defined in \eqref{eq:phi(z)}, associated with ordered, distinct
nodes $\alpha_1<\cdots<\alpha_K$ in $[-1,1]$. The derivative
measures at the encoding points,  $\Lambda_{K,\alpha}^{(r)}$,
for $r=1,2$, satisfy

\vspace{-0.2cm}
\begin{scriptsize}
$$\Lambda_{K,\alpha}^{(1)}\leq2\max_{1\leq j\leq K}
\sum_{\substack{i=1\\i\neq j}}^{K}
\frac{1}{|\alpha_j-\alpha_i|},$$\vspace{-0.3cm}
$$\Lambda_{K,\alpha}^{(2)}\leq
4\max_{1\leq j\leq K}\left[\sum_{\substack{i=1\\i\neq j}}^{K}
\frac{1}{|\alpha_j-\alpha_i|^2}+\left(\sum_{\substack{i=1\\i\neq j}}^{K}
\frac{1}{|\alpha_j-\alpha_i|}\right)^2\right].
$$
\end{scriptsize}
For Chebyshev nodes of the first kind, i.e.,
$\alpha_i=\cos\left(\frac{(2i-1)\pi}{2K}\right)$ for $i\in[K]$,
and $K=LK_\ell$,  these measures satisfy
$\Lambda_{LK_\ell,\alpha}^{(1)}\!\!=\!O(L^2K_\ell^2)$ and
$\Lambda_{LK_\ell,\alpha}^{(2)}\!\!=O(L^4K_\ell^4)$.
\end{lemma}
\subsection{Theoretical Comparison of Both Constructions}
\label{subsec:th discussion}
The approximation-error bound for the Product-based Construction
contains the exponential factor $2^{(L-1)K_\ell}$, as established
in Theorem~\ref{thm:product_error} and
Corollary~\ref{cor:product_K_L}. In contrast,
Lemma~\ref{prop:encoding_point_derivative} establishes polynomial
upper bounds on the derivative measures at the encoding points
for the proposed construction, scaling as $O(L^2K_\ell^2)$ and
$O(L^4K_\ell^4)$ for the first and second derivatives,
respectively. These results provide complementary insights
into the dependence of the two constructions on $L$ and
$K_\ell$. However, the nodal derivative bounds do not establish
polynomial growth of global
approximation error of the proposed construction. We next experimentally evaluate these behaviors for both constructions in terms of relative error as a function of $K_\ell$ and $L$ in the following section.

\section{Experimental Results}
\label{sec:experiments}
We compare the relative errors of the Product-based baseline and the proposed construction for six non-polynomial functions in the BACC setting: $
\psi_1(\mathbf X)=\mathbf X\sin(\mathbf X),
\psi_2(\mathbf X)=(1+e^{-\mathbf X})^{-1},
\psi_3(\mathbf X)=\sin(\mathbf X),
\psi_4(\mathbf X)=\tanh(\mathbf X),
\psi_5(\mathbf X)=\tfrac12\mathbf X\bigl(1+\operatorname{erf}(\mathbf X/\sqrt2)\bigr),
\psi_6(\mathbf X)=\max(\mathbf X,0)+\log(1+e^{-|\mathbf X|}),$
where all nonlinear operations are applied element-wise. For each experiment with $L$ functions, we use $\psi_1,\ldots,\psi_L$. The input matrix entries are drawn independently from the uniform distribution over $[0,1]$, and its rows are equally partitioned into data blocks. We consider two settings: (i) $L=2$, $K=8$, $K_1=K_2=4$, and $\mathbf X\in\mathbb R^{40\times5}$, and (ii) $L=4$, $K=16$, $K_1=K_2=K_3=K_4=4$, and $\mathbf X\in\mathbb R^{80\times5}$. In both settings, each block $\mathbf X_i\in\mathbb R^{5\times5}$, such that each function is assigned four blocks.
The encoding points $\{\alpha_i\}_{i\in[K]}$ and worker evaluation points $\{\beta_i\}_{i\in[N]}$ are chosen Chebyshev nodes of the first and second kinds, respectively. Consecutive encoding points are assigned to function groups, i.e., the first $K_1$ points to $G_1$, the next $K_2$ points to $G_2$, and so on. We use $N=100$ workers and, for each $N_1$, uniformly select $N_1$ responding workers in each of $10^3$ independent trials.
For each $i\in G_\ell$, let $\mathbf Y_i=\psi_\ell(\mathbf X_i)$ denote the centralized output and $\widehat{\mathbf Y}_i$ its approximation. We compute the block-wise squared relative error
$e_{\mathrm{rel},i}
=\frac{\|\mathbf Y_i-\widehat{\mathbf Y}_i\|_F^2}
{\|\mathbf Y_i\|_F^2},$
where $\|\cdot\|_F$ denotes the Frobenius norm. For each $\psi_\ell$, we average this error over its assigned blocks as $e_{\mathrm{rel},\ell}
\!=\!\!\frac{1}{K_\ell}\sum_{i\in G_\ell}e_{\mathrm{rel},i}.$
For each $N_1$, we average $e_{\mathrm{rel},\ell}$ over the $10^3$ trials and report the result on a $\log_{10}(e_{\mathrm{rel},\ell})$ scale. Figs.~\ref{fig:experimental_results}(a) and (b) show that the proposed construction achieves lower relative errors than the Product-based baseline across the tested values of $N_1$.
Next, we investigate the effects of the group size $K_\ell$
and the number of functions $L$ on the relative error for both constructions. In both
experiments, we fix the total number of workers at $N=100$ and
the number of responding workers at $N_1\!=\!40$. We use input
blocks $\mathbf{X}_i\in\mathbb{R}^{5\times5}$ and follow the
same procedure as in Figs.~\ref{fig:experimental_results}(a)
and \ref{fig:experimental_results}(b), averaging the relative
error over $10^3$ trials.
In Fig.~\ref{fig:experimental_results}(c), we fix $L\!=\!4$ and vary $K_\ell$. As observed, the relative error of the proposed construction grows more slowly with increasing $K_\ell$ than that of the Product-based construction. In Fig.~\ref{fig:experimental_results}(d), we plot the relative errors for $\psi_1$ and $\psi_2$, and the experiments for $L=2,4,6$ use the first two, four, and six functions, respectively. Notably, the proposed construction achieves lower relative errors than the Product-based construction across the tested settings. These numerical results are consistent with the trends suggested by Corollary~\ref{cor:product_K_L} and Lemma~\ref{prop:encoding_point_derivative}.

\section{Summary and Future Directions}
\label{sec:summary}
In this work, we proposed a single-round approximate coded computing framework for multiple functions, including non-polynomial functions, and derived approximation-error bounds for both the baseline and the proposed approach in terms of the number of functions and data points per function. Numerical results show lower relative errors than the baseline in the tested settings. Future work includes extending the framework to tolerate Byzantine workers and accelerating distributed multi-task computations, including distributed learning.
\begin{appendix}
\subsection{Proof of Theorem \ref{thm:product_error}}
\label{proof:Theorem:1}
\noindent\textit{Proof.}
We establish the result by applying the approximation-error bound for
Berrut approximated coded computing (BACC) [Theorem~9 in \cite{b6}]
to the joint function $h_{\mathrm P}(z)$ associated with the
product-based construction, defined as
\[h_{\mathrm P}(z)=\sum_{\ell=1}^{L}F_\ell(z)P_\ell(z),
\quad
P_\ell(z)=\prod_{j\notin G_\ell}(z-\alpha_j),
\]
where $d=K-K_\ell=(L-1)K_\ell$ for equal-sized groups. Since
$F_\ell\in C^2([-1,1])$ and $P_\ell$ is a polynomial,
$h_{\mathrm P}\in C^2([-1,1])$. Therefore, to apply the BACC
approximation-error bound, we first establish bounds on the product
masks and their first two derivatives. Since $z,\alpha_j\in[-1,1]$, we have $|z-\alpha_j|\le 2$, and
hence $\|P_\ell\|_\infty\le 2^d$. Moreover, since $P_\ell$ is a
polynomial of degree $d$, Markov's inequalities give
\[\|P_\ell'\|_\infty\leq d^2\|P_\ell\|_\infty\le d^2 2^d\]
and
\[\|P_\ell''\|_\infty\leq \frac{d^2(d^2-1)}{3}\|P_\ell\|_\infty\leq \frac{d^2(d^2-1)}{3}2^d.\]
\noindent Using these bounds, we next derive bounds on the
derivatives of $h_{\mathrm P}$. Differentiating its definition and
applying the product rule, together with the triangle inequality
and the inequality
$\|fg\|_\infty\leq\|f\|_\infty\|g\|_\infty$, yields
\[\begin{aligned}
\|h_{\mathrm P}'\|_\infty
&\leq \sum_{\ell=1}^{L}
\left(\|F_\ell'\|_\infty\|P_\ell\|_\infty
+\|F_\ell\|_\infty\|P_\ell'\|_\infty
\right)\\&\leq 2^d\sum_{\ell=1}^{L}\left(\|F_\ell'\|_\infty+d^2\|F_\ell\|_\infty
\right).\end{aligned}\]
Similarly, 
\[\begin{aligned}\|h_{\mathrm P}''\|_\infty
\leq\sum_{\ell=1}^{L}\big(&\|F_\ell''\|_\infty\|P_\ell\|_\infty
+2\|F_\ell'\|_\infty\|P_\ell'\|_\infty\\&+\|F_\ell\|_\infty\|P_\ell''\|_\infty\big).
\end{aligned}\]
Substituting the bounds on $P_\ell$, $P_\ell'$, and
$P_\ell''$, we obtain
\[\begin{aligned}\|h_{\mathrm P}''\|_\infty
\leq2^d\sum_{\ell=1}^{L}\bigg(&\|F_\ell''\|_\infty+2d^2\|F_\ell'\|_\infty\\&+\frac{d^2(d^2-1)}{3}\|F_\ell\|_\infty\bigg).\end{aligned}\]

We also have
\[\|h_{\mathrm P}\|_\infty\leq 2^d\sum_{\ell=1}^{L}\|F_\ell\|_\infty.\]

We now apply the bound [Theorem 9 \cite{b6}] to the Berrut
reconstruction $r_{\mathrm P}$ of $h_{\mathrm P}$, formed using
the $N-s$ non straggling worker nodes, and define $R_s\triangleq\frac{(s+1)(s+3)\pi^2}{4},
\quad C_{N,s}\triangleq
2(1+R_s)\sin\left(\frac{(s+1)\pi}{2(N-1)}\right).$ Therefore, if $N-s$ is even, we get
\[\|r_{\mathrm P}-h_{\mathrm P}\|_\infty
\leq C_{N,s}\|h_{\mathrm P}''\|_\infty.\]
Substituting the preceding bound on $\|h_{\mathrm P}''\|_\infty$
gives
\[\begin{aligned}\|r_{\mathrm P}-h_{\mathrm P}\|_\infty\leq {}& C_{N,s}2^d
\sum_{\ell=1}^{L}\bigg[\|F_\ell''\|_\infty
+2d^2\|F_\ell'\|_\infty\\&\qquad+
\frac{d^2(d^2-1)}{3}\|F_\ell\|_\infty\bigg].
\end{aligned}\]

Further, if $N-s$ is odd, we get
\[\|r_{\mathrm P}-h_{\mathrm P}\|_\infty\leq C_{N,s}\left(\|h_{\mathrm P}''\|_\infty+
\|h_{\mathrm P}'\|_\infty\right).\]
Substituting the corresponding derivative bounds gives
\[
\begin{aligned}
\|r_{\mathrm P}-h_{\mathrm P}\|_\infty
\leq {}& C_{N,s}2^d
\sum_{\ell=1}^{L}\bigg[
\|F_\ell''\|_\infty
+(2d^2+1)\|F_\ell'\|_\infty\\
&\qquad+
\left(
d^2+\frac{d^2(d^2-1)}{3}
\right)\|F_\ell\|_\infty
\bigg].
\end{aligned}
\]
Finally, suppose that $\max_{\ell\in\{1,\ldots,L\}}
\max_{q\in\{0,1,2\}}
\|F_\ell^{(q)}\|_\infty\le C,$
where $C$ is independent of $K$ and $L$. It follows that
\[\|h_{\mathrm P}'\|_\infty
\leq LC2^d(1+d^2)
\]
and
\[
\|h_{\mathrm P}''\|_\infty
\le LC2^d\left(
1+2d^2+\frac{d^2(d^2-1)}{3}
\right),
\]
while $\|h_{\mathrm P}\|_\infty\le LC2^d$. Therefore, $\|h_{\mathrm P}'\|_\infty=O(Ld^2 2^d),$ and $\|h_{\mathrm P}''\|_\infty=O(Ld^4 2^d)$.
Using $d=(L-1)K_\ell$, we obtain
\[\|h_{\mathrm P}'\|_\infty
=O\left(L(L-1)^2K_\ell^2\,2^{(L-1)K_\ell}
\right)\]
and
\[\|h_{\mathrm P}''\|_\infty
=O\left(L(L-1)^4K_\ell^4\,2^{(L-1)K_\ell}
\right).
\] 
Consequently, for fixed $s$, the derived upper bound on the
reconstruction error scales as
\[
\|r_{\mathrm P}-h_{\mathrm P}\|_\infty
=O\left(
L(L-1)^4 K_\ell^4 2^{(L-1)K_\ell}
\sin\left(\frac{(s+1)\pi}{2(N-1)}\right)
\right).
\]
This completes the proof.
\subsection{Proof of Theorem \ref{thm:global_berrut_error_alpha}}
\label{proof:Theorem:2}
We first derive bounds on the derivatives of the group masks \(M_\ell(z)\), and then use these bounds to bound the derivatives of the joint function \(h_B(z)\), where 
\[h_B(z)=\sum_{\ell=1}^{L}F_\ell(z)M_\ell(z).
\]
By definition,
\[M_\ell(z)=\sum_{i\in G_\ell}\phi_{\alpha_i}(z).\]
Since \(\{G_1,\ldots,G_L\}\) forms a disjoint partition of \([K]\), the triangle inequality gives, for \(r=0,1,2\) and every \(z\in[-1,1]\),
\[
\begin{aligned}
\sum_{\ell=1}^{L}|M_\ell^{(r)}(z)|
&\leq \sum_{\ell=1}^{L}
\sum_{i\in G_\ell}|\phi_{\alpha_i}^{(r)}(z)|=\sum_{i=1}^{K}|\phi_{\alpha_i}^{(r)}(z)|.
\end{aligned}
\]
Taking the supremum over \(z\in[-1,1]\) and recalling the definition
\[\Lambda_K^{(r)}\triangleq\sup_{z\in[-1,1]}
\sum_{i=1}^{K}|\phi_{\alpha_i}^{(r)}(z)|,
\quad r=0,1,2,\]
we obtain
\[\sup_{z\in[-1,1]}\sum_{\ell=1}^{L}|M_\ell^{(r)}(z)|\leq \Lambda_K^{(r)},
\quad r=0,1,2.\tag{1}\]
Further,
\[
h_B'(z)=\sum_{\ell=1}^{L}
\left[F_\ell'(z)M_\ell(z)
+F_\ell(z)M_\ell'(z)\right].
\]
For $r=0,1,2$, define $A_r\triangleq\max_{\ell\in[L]}\|F_\ell^{(r)}\|_\infty,$ where $F_\ell^{(0)}\triangleq F_\ell$ and $\|\cdot\|_\infty$ denotes the supremum norm on $[-1,1]$. Hence, by the triangle inequality and the definitions of \(A_0\) and \(A_1\),
\[\begin{aligned}|h_B'(z)|
&\leq A_1\sum_{\ell=1}^{L}|M_\ell(z)|
+A_0\sum_{\ell=1}^{L}|M_\ell'(z)|.\end{aligned}\]
Taking the supremum over \(z\in[-1,1]\) and applying (1) gives
\[\|h_B'\|_\infty\leq A_1\Lambda_K^{(0)}
+A_0\Lambda_K^{(1)}.\tag{2}\]
Similarly, differentiating a second time gives
\[\begin{aligned}h_B''(z)=\sum_{\ell=1}^{L}\big[
&F_\ell''(z)M_\ell(z)+2F_\ell'(z)M_\ell'(z)\\
&+F_\ell(z)M_\ell''(z)\big].\end{aligned}\]
Therefore,
\[\begin{aligned}|h_B''(z)|
\leq{}&A_2\sum_{\ell=1}^{L}|M_\ell(z)|+2A_1\sum_{\ell=1}^{L}|M_\ell'(z)|\\&+A_0\sum_{\ell=1}^{L}|M_\ell''(z)|.\end{aligned}
\]
Taking the supremum and applying (1) yields
\[\|h_B''\|_\infty\leq A_2\Lambda_K^{(0)}+2A_1\Lambda_K^{(1)}
+A_0\Lambda_K^{(2)}.\tag{3}\]
We now apply the error bound of Theorem 9 of \cite{b6} to the Berrut reconstruction \(r_{\mathrm B}\) of \(h_B\), formed using the \(N-s\) non straggling worker nodes. Therefore, if \(N-s\) is even, then
\[\|r_{\mathrm B}-h_B\|_\infty\leq C_{N,s}\|h_B''\|_\infty.\]
Substituting (3), we obtain
\[\begin{aligned}\|r_{\mathrm B}-h_B\|_\infty
\leq C_{N,s}\big[&A_2\Lambda_K^{(0)}
+2A_1\Lambda_K^{(1)}\\&+A_0\Lambda_K^{(2)}\big],
\end{aligned}\]
Further, if \(N-s\) is odd,
\[\|r_{\mathrm B}-h_B\|_\infty\leq C_{N,s}\left(
\|h_B''\|_\infty+\|h_B'\|_\infty\right).\]
Substituting (2) and (3) gives
\[\begin{aligned}\|r_{\mathrm B}-h_B\|_\infty\leq C_{N,s}\big[
&A_2\Lambda_K^{(0)}+2A_1\Lambda_K^{(1)}
+A_0\Lambda_K^{(2)}\\&+A_1\Lambda_K^{(0)}+A_0\Lambda_K^{(1)}\big],
\end{aligned}\]
where $C_{N,s}\triangleq
2(1+R_s)\sin\left(\frac{(s+1)\pi}{2(N-1)}\right)$. This completes the proof.
\subsection{Proof of Lemma \ref{prop:encoding_point_derivative}}
\label{proof:prop1}
Consider $K$ distinct encoding nodes $\alpha_1,\ldots,\alpha_K\in[-1,1]$ chosen as the Chebyshev nodes of the first kind, and the Berrut weights $w_i=(-1)^i$. The corresponding Berrut basis functions are defined as
\begin{equation}
\phi_i(z)
=\frac{\displaystyle \frac{w_i}{z-\alpha_i}}
{\displaystyle \sum_{j=1}^{K}\frac{w_j}{z-\alpha_j}},
\qquad i=1,\ldots,K.
\label{eq:berrut_basis}
\end{equation}
We derive bounds on the first and second derivatives of these basis functions evaluated at the encoding nodes $\alpha_{k}$ for $k=1,2,\ldots K$. To this end, we define \[D(z)=\sum_{j=1}^{K}\frac{w_j}{z-\alpha_j}.\]
For a fixed index $i$, we separate the $i$th term as
\[D(z)=\frac{w_i}{z-\alpha_i}+D_i(z),\quad D_i(z)=\sum_{\substack{j=1\\j\neq i}}^{K}\frac{w_j}{z-\alpha_j}.\]
Substitution into \eqref{eq:berrut_basis} gives
\[\phi_i(z)=\frac{w_i}{w_i+(z-\alpha_i)D_i(z)}.\]
Define $q_i(z)=(z-\alpha_i)D_i(z)$. Then
\[\phi_i(z)=\frac{w_i}{w_i+q_i(z)}.\]
Since $q_i(\alpha_i)=0$ and $w_i^2=1$, this representation also gives
the cardinality property $\phi_i(\alpha_i)=1$. Differentiating with respect to $z$ gives
\begin{equation}
\label{eq:first deriv}
    \phi_i'(z)=-\frac{w_iq_i'(z)}{[w_i+q_i(z)]^2}.
\end{equation}
Furthermore,
\[q_i'(z)=D_i(z)+(z-\alpha_i)D_i'(z),\]
where
\[D_i'(z)=-\sum_{\substack{j=1\\j\neq i}}^{K}\frac{w_j}{(z-\alpha_j)^2}.\]
Hence, at $z=\alpha_i$, the second term in $q_i'(z)$ vanishes and
\[q_i'(\alpha_i)=\sum_{\substack{j=1\\j\neq i}}^{K}\frac{w_j}{\alpha_i -\alpha_j}.\]
Using $q_i(\alpha_i)=0$ and $w_i^2=1$, we obtain
\begin{equation}
\phi_i'(\alpha_i)=-w_i\sum_{\substack{j=1\\j\neq i}}^{K}
\frac{w_j}{\alpha_i-\alpha_j}.
\label{eq:phi_prime_diagonal}
\end{equation}
Since $|w_i|=|w_j|=1$, this implies
\[\left|\phi_i'(\alpha_i)\right|\leq\sum_{\substack{j=1\\j\neq i}}^{K}\frac{1}{|\alpha_i-\alpha_j|}.\]
For the off-diagonal case, let $i\neq j$. Separating the term
corresponding to $\alpha_j$ in the denominator of
\eqref{eq:berrut_basis} i.e.,  $\phi_i(z)
=\frac{\displaystyle \frac{w_i}{z-\alpha_i}}
{\displaystyle \sum_{j=1}^{K}\frac{w_j}{z-\alpha_j}},$ and multiplying the numerator and denominator by $(z-\alpha_j)$ gives
\[\phi_i(z)=\frac{\displaystyle
\frac{w_i(z-\alpha_j)}{z-\alpha_i}}{\displaystyle w_j+
(z-\alpha_j)\sum_{\substack{k=1\\k\neq j}}^{K}\frac{w_k}{z-\alpha_k}}.\]
The denominator tends to $w_j\neq0$ as $z\rightarrow\alpha_j$, while
the numerator has a simple zero at $z=\alpha_j$. Thus,
$\phi_i(\alpha_j)=0$. To evaluate the derivative, define
\[N(z)=\frac{w_i(z-\alpha_j)}{z-\alpha_i},\quad Q(z)=w_j+(z-\alpha_j)
\sum_{\substack{k=1\\k\neq j}}^{K}\frac{w_k}{z-\alpha_k},\]
so that $\phi_i(z)=N(z)/Q(z)$. At $z=\alpha_j$,
\[N(\alpha_j)=0,\quad Q(\alpha_j)=w_j.\]
Therefore, by the quotient rule,
\[\phi_i'(\alpha_j)=\frac{N'(\alpha_j)Q(\alpha_j)-N(\alpha_j)Q'(\alpha_j)} {Q(\alpha_j)^2}=\frac{N'(\alpha_j)}{w_j}.\]
Since
\[N'(z)=w_i\frac{(z-\alpha_i)-(z-\alpha_j)}{(z-\alpha_i)^2},\]
we have
\[N'(\alpha_j)=\frac{w_i}{\alpha_j-\alpha_i}.\]
Hence,
\begin{equation}
\phi_i'(\alpha_j)=\frac{w_i}{w_j(\alpha_j-\alpha_i)},\quad i\neq j.
\label{eq:phi_prime_offdiagonal}
\end{equation}
Since $|w_i|=|w_j|=1$, it follows that
\[\left|\phi_i'(\alpha_j)\right|=\frac{1}{|\alpha_j-\alpha_i|},\quad i\neq j.\]
Combining the diagonal and off-diagonal contributions, for any fixed
interpolation node $\alpha_j$,
\[
\begin{aligned}
\sum_{i=1}^{K}\left|\phi_i'(\alpha_j)\right|&=
\left|\phi_j'(\alpha_j)\right|+\sum_{\substack{i=1\\i\neq j}}^{K}
\left|\phi_i'(\alpha_j)\right|\\
&\leq \sum_{\substack{i=1\\i\neq j}}^{K}
\frac{1}{|\alpha_j-\alpha_i|}+
\sum_{\substack{i=1\\i\neq j}}^{K}
\frac{1}{|\alpha_j-\alpha_i|}\\
&=2\sum_{\substack{i=1\\i\neq j}}^{K}
\frac{1}{|\alpha_j-\alpha_i|}.
\end{aligned}\]
Therefore, defining the nodal first-derivative sensitivity by
\[\Lambda_{K,\alpha}^{(1)}=
\max_{1\leq j\leq K}\sum_{i=1}^{K}\left|\phi_i'(\alpha_j)\right|,\]
we obtain the bound
\begin{equation}
\Lambda_{K,\alpha}^{(1)}
\leq 2\max_{1\leq j\leq K}
\sum_{\substack{i=1\\i\neq j}}^{K}
\frac{1}{|\alpha_j-\alpha_i|}.
\label{eq:nodal_lambda_first_bound}
\end{equation}
Now, we now derive a bound on the sum of the absolute second derivatives of the Berrut basis functions at each encoding point. From \eqref{eq:first deriv}, we have
\[\phi_i'(z)=-\frac{w_iq_i'(z)}{[w_i+q_i(z)]^2}.\]
Differentiating the above expression once more gives
\[\phi_i''(z)=-\frac{w_iq_i''(z)}{\left[w_i+q_i(z)\right]^2}
+\frac{2w_i[q_i'(z)]^2}{\left[w_i+q_i(z)\right]^3}.\]
From \[q_i'(z)=D_i(z)+(z-\alpha_i)D_i'(z),\]
we obtain
\[q_i''(z)=2D_i'(z)+(z-\alpha_i)D_i''(z).\]
Since
\[D_i'(z)=-\sum_{\substack{j=1\\j\neq i}}^{K}\frac{w_j}{(z -\alpha_j)^2},\]
it follows that
\[q_i''(\alpha_i)=-2\sum_{\substack{j=1\\j\neq i}}^{K}\frac{w_j}{(\alpha_i-\alpha_j)^2}.\]
Using $q_i(\alpha_i)=0$ and $w_i^2=1$, the second derivative at the
interpolation node $\alpha_i$ is therefore
\begin{equation}
\phi_i''(\alpha_i)=2w_i
\sum_{\substack{j=1\\j\neq i}}^{K}
\frac{w_j}{(\alpha_i-\alpha_j)^2}
+2\left(\sum_{\substack{j=1\\j\neq i}}^{K}\frac{w_j}{\alpha_i-\alpha_j}
\right)^2.
\label{eq:phi_second_node}
\end{equation}
Consequently,
\[\left|\phi_i''(\alpha_i)\right|\leq 2
\sum_{\substack{j=1\\j\neq i}}^{K}
\frac{1}{|\alpha_i-\alpha_j|^2}+2\left(\sum_{\substack{j=1\\j\neq i}}^{K}
\frac{1}{|\alpha_i-\alpha_j|}\right)^2.\]

For the off-diagonal case, let $i\neq j$. As in the first-derivative
analysis, we separate the term corresponding to $\alpha_j$ in the
denominator of \eqref{eq:berrut_basis}. Multiplying the numerator and
denominator by $(z-\alpha_j)$ gives
\[\phi_i(z)=\frac{\displaystyle\frac{w_i(z-\alpha_j)}{z-\alpha_i}}
{\displaystyle w_j+(z-\alpha_j)D_j(z)},\]
where
\[D_j(z)=\sum_{\substack{k=1\\k\neq j}}^{K}\frac{w_k}{z-\alpha_k}.\]
Define
\[N_{ij}(z)=\frac{w_i(z-\alpha_j)}{z-\alpha_i},\quad Q_j(z)
=w_j+(z-\alpha_j)D_j(z),\]
so that
\[\phi_i(z)=\frac{N_{ij}(z)}{Q_j(z)}.\]
Since $i\neq j$, both $N_{ij}(z)$ and $Q_j(z)$ are regular at
$z=\alpha_j$, and
\[N_{ij}(\alpha_j)=0,\quad Q_j(\alpha_j)=w_j.\]

We now differentiate the numerator. Since
\[N_{ij}(z)=w_i\frac{z-\alpha_j}{z-\alpha_i},\]
the first derivative is
\[N_{ij}'(z)=w_i\frac{(z-\alpha_i)-(z-\alpha_j)}{(z-\alpha_i)^2}
=w_i\frac{\alpha_j-\alpha_i}{(z-\alpha_i)^2}.\]
Hence,
\[N_{ij}'(\alpha_j)=\frac{w_i}{\alpha_j-\alpha_i}.\]
Differentiating once more gives
\[N_{ij}''(z)=-2w_i\frac{\alpha_j-\alpha_i}{(z-\alpha_i)^3},\]
and therefore
\[N_{ij}''(\alpha_j)=-\frac{2w_i}{(\alpha_j-\alpha_i)^2}.\]
Next, differentiating $Q_j(z)$ gives
\[Q_j'(z)=D_j(z)+(z-\alpha_j)D_j'(z).\]
Thus,
\[Q_j'(\alpha_j)=D_j(\alpha_j)=\sum_{\substack{k=1\\k\neq j}}^{K}\frac{w_k}{\alpha_j-\alpha_k}.\]
For the second derivative of the quotient
$\phi_i(z)=N_{ij}(z)/Q_j(z)$, we use
\[\left(\frac{N}{Q}\right)''=\frac{N''}{Q}-\frac{2N'Q'}{Q^2}-\frac{NQ''}{Q^2}+\frac{2N(Q')^2}{Q^3}.\]
Since $N_{ij}(\alpha_j)=0$, the last two terms vanish at
$z=\alpha_j$. Therefore,
\[\phi_i''(\alpha_j)=\frac{N_{ij}''(\alpha_j)}{Q_j(\alpha_j)}-
\frac{2N_{ij}'(\alpha_j)Q_j'(\alpha_j)}{Q_j(\alpha_j)^2}.\]
Substituting the expressions above and using
$Q_j(\alpha_j)=w_j$ gives
\[\phi_i''(\alpha_j)=-\frac{2w_i}{w_j(\alpha_j-\alpha_i)^2}-\frac{2w_i}
{w_j^2(\alpha_j-\alpha_i)}\sum_{\substack{k=1\\k\neq j}}^{K}
\frac{w_k}{\alpha_j-\alpha_k}.\]
Since $w_j^2=1$, we obtain
\begin{equation}
\phi_i''(\alpha_j)=-2\frac{w_i}{w_j(\alpha_j-\alpha_i)^2}-
2\frac{w_i}{\alpha_j-\alpha_i}\sum_{\substack{k=1\\k\neq j}}^{K}
\frac{w_k}{\alpha_j-\alpha_k},\quad i\neq j.
\label{eq:phi_second_offdiag}
\end{equation}
Since $|w_i|=|w_j|=1$, applying the triangle inequality yields
\[\left|\phi_i''(\alpha_j)\right|\leq\frac{2}{|\alpha_j-\alpha_i|^2}+\frac{2}{|\alpha_j-\alpha_i|}\sum_{\substack{k=1\\k\neq j}}^{K}\frac{1}{|\alpha_j-\alpha_k|},\quad i\neq j.\]
Combining the diagonal and off-diagonal contributions, for a fixed
interpolation node $\alpha_j$ we have
\[\begin{aligned}
\sum_{i=1}^{K}\left|\phi_i''(\alpha_j)\right|&=\left|\phi_j''(\alpha_j)\right|
+\sum_{\substack{i=1\\i\neq j}}^{K}\left|\phi_i''(\alpha_j)\right|\\&\leq
2\sum_{\substack{i=1\\i\neq j}}^{K}\frac{1}{|\alpha_j-\alpha_i|^2}+
2\left(\sum_{\substack{i=1\\i\neq j}}^{K}\frac{1}{|\alpha_j-\alpha_i|}\right)^2+\\ & 2\sum_{\substack{i=1\\i\neq j}}^{K}
\frac{1}{|\alpha_j-\alpha_i|^2}+2\left(\sum_{\substack{i=1\\i\neq j}}^{K}
\frac{1}{|\alpha_j-\alpha_i|}\right)\\ &\left(\sum_{\substack{k=1\\k\neq j}}^{K}
\frac{1}{|\alpha_j-\alpha_k|}\right).
\end{aligned}\]
The two sums appearing in the last product are identical, and hence
\begin{equation}
\sum_{i=1}^{K}
\left|\phi_i''(\alpha_j)\right|\leq 4\sum_{\substack{i=1\\i\neq j}}^{K}
\frac{1}{|\alpha_j-\alpha_i|^2}+4\left(\sum_{\substack{i=1\\i\neq j}}^{K}
\frac{1}{|\alpha_j-\alpha_i|}\right)^2.
\label{eq:nodal_phi_second_sum}
\end{equation}
Define the nodal second-derivative sensitivity as
\[\Lambda_{K,\alpha}^{(2)}=\max_{1\leq j\leq K}\sum_{i=1}^{K}\left|\phi_i''(\alpha_j)\right|.\]
It then follows from \eqref{eq:nodal_phi_second_sum} that
\begin{equation}
\Lambda_{K,\alpha}^{(2)}\leq 4\max_{1\leq j\leq K}\left[\sum_{\substack{i=1\\i\neq j}}^{K}\frac{1}{|\alpha_j-\alpha_i|^2}+\left(\sum_{\substack{i=1\\i\neq j}}^{K}\frac{1}{|\alpha_j-\alpha_i|}\right)^2\right].
\label{eq:nodal_lambda_second_bound}
\end{equation}
We now specialize the bounds in
\eqref{eq:nodal_lambda_first_bound} and
\eqref{eq:nodal_lambda_second_bound} to the first-kind Chebyshev
nodes, which are used as encoding points in our construction. For first-kind Chebyshev nodes, let
\[
\alpha_i=\cos\theta_i,\qquad
\theta_i=\frac{(2i-1)\pi}{2K}.
\]
To determine the dependence on $K$, we define
\[S_j^{(1)}=\sum_{i\neq j}\frac{1}{|\alpha_j-\alpha_i|},
\qquad S_j^{(2)}=\sum_{i\neq j}\frac{1}{|\alpha_j-\alpha_i|^2}.\]
Using the identity for the difference of two cosines, we obtain
\[|\alpha_i-\alpha_j|=2\sin\left(\frac{(i+j-1)\pi}{2K}\right)
\sin\left(\frac{|i-j|\pi}{2K}\right).\]
For $x\in[0,\pi]$, the inequality
\[\sin x\geq\frac{2}{\pi}\min\{x,\pi-x\}\]
implies that, for $i\neq j$,
\[|\alpha_i-\alpha_j|\geq
\frac{2|i-j|}{K^2}\min\{i+j-1,\,2K-i-j+1\}.\]
Since
\[\min\{i+j-1,\,2K-i-j+1\}\geq |i-j|,\]
we obtain the uniform lower bound
\[|\alpha_i-\alpha_j|\geq\frac{2|i-j|^2}{K^2}.\]
Consequently, for every $j\in[K]$,
\[\begin{aligned}S_j^{(1)}
&\leq\frac{K^2}{2}\sum_{i\neq j}\frac{1}{|i-j|^2}\\&\leq K^2\sum_{m=1}^{K-1}\frac{1}{m^2}=O(K^2),
\end{aligned}\]
where the last bound follows because
$\sum_{m=1}^{\infty}m^{-2}$ converges. Hence,
\[\max_{1\leq j\leq K}S_j^{(1)}=O(K^2).\]
Combining this result with
\eqref{eq:nodal_lambda_first_bound} gives $\Lambda_{K,\alpha}^{(1)}=O(K^2).$ For the second derivative, the same uniform distance bound gives
\[\begin{aligned}S_j^{(2)}
&\leq\frac{K^4}{4}\sum_{i\neq j}\frac{1}{|i-j|^4}\\&\leq\frac{K^4}{2}
\sum_{m=1}^{K-1}\frac{1}{m^4}=O(K^4),
\end{aligned}\]
uniformly over $j$, since
$\sum_{m=1}^{\infty}m^{-4}$ converges. Moreover,
\[\max_{1\leq j\leq K}\bigl(S_j^{(1)}\bigr)^2=O(K^4).\]
Substituting these bounds into
\eqref{eq:nodal_lambda_second_bound} gives $\Lambda_{K,\alpha}^{(2)}=O(K^4)$ Further, for equal-sized groups, where $K=LK_\ell$, the above bounds gives $\Lambda_{LK_\ell,\alpha}^{(1)}=O(L^2K_\ell^2),\quad\Lambda_{LK_\ell,\alpha}^{(2)}=O(L^4K_\ell^4).$
This completes the proof.
\end{appendix}

\begin{thebibliography}{16}

\bibitem{a1} V. K. Garg, \emph{Elements of distributed computing}, in John Wiley and Sons, 2002.

\bibitem{c1} Q. Yu, M. A. Maddah-Ali and A. S. Avestimehr, ``Straggler mitigation in distributed matrix multiplication: Fundamental limits and optimal coding," in \emph{IEEE Transactions on Information Theory}, vol. 66, no. 3, pp. 1920--1933, March 2020.

\bibitem{w1}A. Ramamoorthy, A. B. Das and L. Tang, ``Straggler-esistant distributed matrix computation via coding theory: Removing a bottleneck in large-scale data processing," in \emph{IEEE Signal Processing Magazine}, vol. 37, no. 3, pp. 136-145, May 2020.


\bibitem{w3} J. Zhu and S. Li, ``Generalized Lagrange coded computing: A flexible computation-communication tradeoff," in \emph{IEEE International Symposium on Information Theory (ISIT)}, Finland, 2022, pp. 832--837.


\bibitem{a11}  M. Soleymani, H. Mahdavifar and A. S. Avestimehr, ``Analog privacy-preserving coded computing," in \emph{IEEE International Symposium on Information Theory}, Melbourne, Australia, 2021, pp. 1865-1870.

\bibitem{a12} Mahdi Soleymani, Hessam Mahdavifar, A. Salman Avestimehr, ``Privacy-preserving distributed learning in the analog domain," in \emph{arXiv preprint arXiv:2007.08803}, 2020.

 \bibitem{f1} M. Fahim and V. R. Cadambe, ``Numerically stable polynomially coded computing," in \emph{IEEE Transactions on Information Theory}, vol. 67, no. 5, pp. 2758-2785, May 2021.


\bibitem{a8}C. S. Yang and A. S. Avestimehr, ``Coded computing for secure Boolean computations," in \emph{IEEE Journal on Selected Areas in Information Theory}, vol. 2, no. 1, pp. 326-337, March 2021.

\bibitem{a9} Yang et al., ``Secure coded computation for efficient distributed learning in mobile IoT," in \emph{ 2021 18th Annual IEEE SECON}, 2021.


\bibitem{a2} R. Borah and J. Harshan, ``On securing analog Lagrange coded computing from colluding adversaries,'' in \emph{2024 IEEE International Symposium on Information Theory (ISIT)}, Athens, Greece, 2024.

\bibitem{SBACC} R. Borah and J. Harshan, ``On securing Berrut approximated coded computing through discrete cosine transforms,” in  \emph{IEEE Information Theory Workshop (ITW)}, Sydney, Australia, 2025.

\bibitem{a10} M. Kim and J. Lee, ``Private secure coded computation," in \emph{IEEE ISIT}, Paris, France, 2019, pp. 1097-1101.

\bibitem{w2} Q. Yu, S. Li, N. Raviv, S. M. M. Kalan, M. Soltanolkotabi, and S. A. Avestimehr, ``Lagrange coded computing: Optimal design for resiliency, security, and privacy,” in \emph{22nd Int. Conf. Artif. Intell. Stat.}, 2019, pp. 1215--1225.

\bibitem{b5} M. Soleymani, H. Mahdavifar and A. S. Avestimehr, ``Analog Lagrange coded computing," in \emph{IEEE Journal on Selected Areas in Information Theory}, vol. 2, no. 1, pp. 283--295, March 2021.

\bibitem{tandon2017gradient} R. Tandon, Q. Lei, A. G. Dimakis, and N. Karampatziakis, ``Gradient Coding: Avoiding Stragglers in Distributed Learning,'' in \emph{Proceedings of the 34th International Conference on Machine Learning}, vol. 70, pp. 3368--3376, 2017.

\bibitem{lee2018speeding} K. Lee, M. Lam, R. Pedarsani, D. Papailiopoulos, and K. Ramchandran, ``Speeding Up Distributed Machine Learning Using Codes,'' in \emph{IEEE Transactions on Information Theory}, 2018.

\bibitem{b6}
T. Jahani-Nezhad and M. A. Maddah-Ali, ``Berrut Approximated Coded Computing: Straggler Resistance Beyond Polynomial Computing,'' \emph{IEEE Transactions on Pattern Analysis and Machine Intelligence}, vol. 45, no. 1, pp. 111--122, Jan. 2023.

\bibitem{caruana1997multitask} R. Caruana, ``Multitask Learning," in \emph{Machine Learning}, vol. 28, no. 1, pp. 41--75, 1997.
\bibitem{zhang2022survey} Y. Zhang and Q. Yang, ``A Survey on Multi-Task Learning,''
in \emph{IEEE Transactions on Knowledge and Data Engineering},
vol. 34, no. 12, pp. 5586--5609, Dec. 2022.


\bibitem{main} W. Kim, S. Kruglik and H. Mao Kiah, ``Verifiable Coded Computation of Multiple Functions,'' in \emph{IEEE Transactions on Information Forensics and Security}, vol. 19, pp. 8009-8022, 2024.


\bibitem{fullversion}
R. Borah, W. Kim, J. Harshan, S. Kruglik, and H. M. Kiah,
``Multiple-Function Approximate Coded Computing: Constructions and Error Bounds,''
arXiv preprint arXiv:2610.16549, 2026.

\end{thebibliography}
\end{document}